\documentclass{amsart}
\usepackage{hyperref}
\usepackage{curves}

\newcommand{\arxiv}[1]{\href{http://arxiv.org/abs/#1}{arXiv:#1}}
\newcommand*{\mailto}[1]{\href{mailto:#1}{\nolinkurl{#1}}}

\newtheorem{theorem}{Theorem}[section]
\newtheorem{lemma}[theorem]{Lemma}
\newtheorem{corollary}[theorem]{Corollary}

\newtheorem{hypothesis}[theorem]{Hypothesis {\bf H.}\hspace*{-0.6ex}}

\newcommand{\R}{\mathbb{R}}
\newcommand{\N}{\mathbb{N}}
\newcommand{\Z}{\mathbb{Z}}
\newcommand{\C}{\mathbb{C}}
\newcommand{\T}{\mathbb{T}}

\newcommand{\nn}{\nonumber}
\newcommand{\be}{\begin{equation}}
\newcommand{\ee}{\end{equation}}
\newcommand{\bea}{\begin{eqnarray}}
\newcommand{\eea}{\end{eqnarray}}

\newcommand{\ol}{\overline}
\newcommand{\ti}{\tilde}
\newcommand{\wti}{\widetilde}

\newcommand{\id}{\mathbb{I}}
\newcommand{\I}{\mathrm{i}}
\newcommand{\E}{\mathrm{e}}
\newcommand{\ind}{\mathrm{ind}}
\newcommand{\re}{\mathrm{Re}}

\DeclareMathOperator{\res}{Res}
\newcommand{\lz}{\ell^2(\Z)}

\def\Xint#1{\mathchoice
   {\XXint\displaystyle\textstyle{#1}}%
   {\XXint\textstyle\scriptstyle{#1}}%
   {\XXint\scriptstyle\scriptscriptstyle{#1}}%
   {\XXint\scriptscriptstyle\scriptscriptstyle{#1}}%
   \!\int}
\def\XXint#1#2#3{{\setbox0=\hbox{$#1{#2#3}{\int}$}
     \vcenter{\hbox{$#2#3$}}\kern-.5\wd0}}
\def\dashint{\Xint-}

\newcommand{\eps}{\varepsilon}

\newcommand{\lam}{\lambda}
\newcommand{\gam}{\gamma}

\numberwithin{equation}{section}

\newcommand{\sigI}{\begin{pmatrix} 0 & 1 \\ 1 & 0 \end{pmatrix}}
\newcommand{\ssigI}{\big(\begin{smallmatrix} 0 & 1 \\ 1 & 0\end{smallmatrix}\big)}
\newcommand{\rI}{\begin{pmatrix}  1 & 1 \end{pmatrix}}
\newcommand{\rN}{\begin{pmatrix}  0 & 0 \end{pmatrix}}

\begin{document}

\title[Long-Time Asymptotics for the Toda Lattice]{Long-Time Asymptotics for the Toda Lattice in the Soliton Region}

\author[H. Kr\"uger]{Helge Kr\"uger}
\address{Department of Mathematics\\ Rice University\\ Houston\\ TX 77005\\ USA}
\email{\mailto{helge.krueger@rice.edu}}
\urladdr{\url{http://math.rice.edu/~hk7/}}

\author[G. Teschl]{Gerald Teschl}
\address{Faculty of Mathematics\\
Nordbergstrasse 15\\ 1090 Wien\\ Austria\\ and International Erwin Schr\"odinger
Institute for Mathematical Physics, Boltzmanngasse 9\\ 1090 Wien\\ Austria}
\email{\mailto{Gerald.Teschl@univie.ac.at}}
\urladdr{\url{http://www.mat.univie.ac.at/~gerald/}}

\thanks{Research supported by the Austrian Science Fund (FWF) under Grant No.\ Y330.}
\thanks{Math. Z. {\bf 262}, 585--602 (2009)}

\keywords{Riemann--Hilbert problem, Toda lattice, solitons}
\subjclass[2000]{Primary 37K40, 37K45; Secondary 35Q15, 37K10}

\begin{abstract}
We apply the method of nonlinear steepest descent to compute the long-time
asymptotics of the Toda lattice for decaying initial data in the soliton region. In addition,
we point out how to reduce the problem in the remaining region to the known case without solitons.
\end{abstract}

\maketitle

\section{Introduction}

In this paper we want to compute the long time asymptotics for
the doubly infinite Toda lattice which reads in Flaschka's
variables (see e.g.\ \cite{tjac}, \cite{taet}, or \cite{ta})
\be \label{tl}
\aligned
\dot b(n,t) &= 2(a(n,t)^2 -a(n-1,t)^2),\\
\dot a(n,t) &= a(n,t) (b(n+1,t) -b(n,t)),
\endaligned
\ee
$(n,t) \in \Z \times \R$. Here the dot denotes differentiation with
respect to time. We will consider solutions $(a,b)$ satisfying
\be \label{decay}
\sum_n |n|^l (|a(n,t) - \frac{1}{2}| + |b(n,t)|) < \infty
\ee
for every $l\in\N$ for one (and hence for all, see \cite{tjac}) $t\in\R$. It is
well-known that this initial value problem has unique global solutions
which can be computed via the inverse scattering transform \cite{tjac}.

The long-time asymptotics for this problem were first given by Novokshenov and
Habibullin \cite{nh} and were later made rigorous by Kamvissis \cite{km}, however,
only in the case without solitons.
The purpose of the present paper is to finally fill this gap and show how to include solitons.
As in \cite{km}, our approach is based on the nonlinear steepest
descent analysis for oscillatory Riemann--Hilbert problems from Deift and Zhou \cite{dz}.
It turns out that in the case of solitons, two new phenomena enter the scene
which require significant adaptions to the original method of Deift and Zhou.
Of course our technique also applies to other soliton equations, e.g., the Korteweg--de Vries equation.

First of all, it is well-known that there is a subtle nonuniqueness issue for the involved Riemann--Hilbert problems
(see e.g.\ \cite[Chap.~38]{bdt}).
In fact, in certain exceptional sets the corresponding vanishing Riemann--Hilbert problem has a nontrivial solution
and hence by Fredholm theory the corresponding matrix Riemann--Hilbert problem has no solution at all.
This problem does not affect the similarity region, since it is easy to see that there it
does not happen for sufficiently large times. However, in the soliton region,
this occurs precisely in the neighborhoods of the single solitons. To avoid this problem
we will work directly with the vector Riemann--Hilbert problem and impose a symmetry condition in order to
ensure uniqueness. This also has the advantage that it eliminates the step of going forth
and back between the vector and matrix Riemann--Hilbert problem. It should be pointed out here that even the symmetry
condition alone does not guarantee uniqueness, rather we will need existence of a certain
solution with an additional property. We will also demonstrate that this additional property is
in fact necessary.

Secondly, in the regions of the single solitons, the zeroth order asymptotics are not equal
to zero but given by a one soliton solution. Hence for the usual perturbation argument based
on the second resolvent identity to work, a uniform bound for the inverse of the singular integral equation
associated with the one soliton solution is needed. Unfortunately, such a bound
cannot be easily obtained. To overcome this problem we will shift the leading
asymptotics from the one soliton solution to the inhomogeneous part of the singular integral
equation and craft our Cauchy kernel in such a way that it preserves the pole conditions for this
single soliton.

To state our main result, we begin by recalling that the sequences $a(n,t)$, $b(n,t)$, $n\in\Z$,
for fixed $t\in\R$, are uniquely determined by their scattering data, that is, by the right reflection
coefficient $R_+(z,t)$, $|z|=1$ and the eigenvalues $\lam_j\in(-\infty,-1)\cup(1,\infty)$, $j=1,\dots, N$
together with the corresponding right norming constants $\gam_{+,j}(t)>0$, $j=1,\dots, N$.
Rather than in the complex plane, we will work on the unit disc using the usual
Joukowski transformation
\be\label{defzlam}
\lam = \frac{1}{2} \left(z + \frac{1}{z}\right),\quad z= \lam - \sqrt{\lam^2 -1}, \qquad
\lam\in\C, \: |z|\leq 1.
\ee
In these new coordinates the eigenvalues $\lam_j\in(-\infty,-1)\cup(1,\infty)$ will be denoted by
$\zeta_j\in(-1,0)\cup(0,1)$. The continuous spectrum $[-1,1]$ is mapped to the unit circle
$\T$. Moreover, the phase of the associated Riemann--Hilbert problem is given by
\begin{equation} \label{eq:Phi}
\Phi(z)=z-z^{-1}+2 \frac{n}{t} \log(z).
\end{equation}
and the stationary phase points, $\Phi'(z)=0$, are denoted by
\be
z_0=  -\frac{n}{t} - \sqrt{(\frac{n}{t})^2 -1}, \quad z_0^{-1}= -\frac{n}{t} + \sqrt{(\frac{n}{t})^2 -1},
\qquad \lam_0=-\frac{n}{t}.
\ee
For $\frac{n}{t}<-1$ we have $z_0\in(0,1)$, for $-1\le \frac{n}{t} \le1$
we have $z_0\in\T$ (and hence $z_0^{-1}=\ol{z_0}$), and for $\frac{n}{t}>1$
we have $z_0\in(-1,0)$. For $|\frac{n}{t}|>1$ we will also need the value
$\zeta_0\in(-1,0)\cup(0,1)$ defined via $\re(\Phi(\zeta_0))=0$, that is,
\be
\frac{n}{t} = -\frac{\zeta_0 - \zeta_0^{-1}}{2\log(|\zeta_0|)}.
\ee
We will set $\zeta_0=-1$ if $|\frac{n}{t}|\le 1$ for notational convenience.
A simple analysis shows that for $\frac{n}{t}<-1$ we have $0<\zeta_0 < z_0 <1$ and
for $\frac{n}{t}>1$ we have $-1<z_0 < \zeta_0 <0$.

Furthermore, recall that the transmission coefficient $T(z)$, $|z|\le 1$, is time independent
and can be reconstructed using the Poisson--Jensen formula. In particular, we define
the partial transmission coefficient with respect to $z_0$ by
\begin{align}\nn
T(z,z_0) &=\\ \label{def:Tzz0}
& \begin{cases}
\prod\limits_{\zeta_k\in(\zeta_0,0)} |\zeta_k| \frac{z-\zeta_k^{-1}}{z-\zeta_k}, & z_0 \in (-1,0), \\
\left(\prod\limits_{\zeta_k\in(-1,0)} |\zeta_k| \frac{z-\zeta_k^{-1}}{z-\zeta_k} \right)
\exp\left(\frac{1}{2\pi\I}\int\limits_{\ol{z_0}}^{z_0}\log(|T(s)|) \frac{s+z}{s-z} \frac{ds}{s}\right), & |z_0| = 1, \\
\left(\prod\limits_{\zeta_k\in(-1,0)\cup(\zeta_0,1)}\!\!\! |\zeta_k| \frac{z-\zeta_k^{-1}}{z-\zeta_k} \right)
\exp\left(\frac{1}{2\pi\I}\int\limits_{\T}\log(|T(s)|) \frac{s+z}{s-z} \frac{ds}{s}\right), & z_0 \in (0,1).
\end{cases}
\end{align}
Here, in the case $z_0\in\T$,  the integral is to be taken along the arc $\Sigma(z_0)= \{z \in\T | \re(z)<\re(z_0)\}$
oriented counterclockwise. For $z_0\in(-1,0)$ we set $\Sigma(z_0)=\emptyset$ and for $z_0\in(0,1)$ we
set $\Sigma(z_0)=\T$. Then $T(z,z_0)$ is meromorphic for $z\in\C\backslash\Sigma(z_0)$.
Observe $T(z,z_0)=T(z)$ once $z_0\in(0,1)$ and $(0,\zeta_0)$ contains no eigenvalues. Moreover,
$T(z,z_0)$ can be computed in terms of the scattering data since $|T(z)|^2= 1- |R_+(z,t)|^2$.

Moreover, we set
\begin{align}\nn
T_0(z_0) &= T(0,z_0)\\
&= \begin{cases}
\prod\limits_{\zeta_k\in(\zeta_0,0)} |\zeta_k|^{-1}, & z_0 \in (-1,0),\\
\left(\prod\limits_{\zeta_k\in(-1,0)} |\zeta_k|^{-1} \right)
\exp\left(\frac{1}{2\pi\I}\int\limits_{\ol{z_0}}^{z_0}\log(|T(s)|) \frac{ds}{s}\right), & |z_0| = 1, \\
\left(\prod\limits_{\zeta_k\in(-1,0)\cup(\zeta_0,1)} |\zeta_k|^{-1} \right)
\exp\left(\frac{1}{2\pi \I}\int\limits_{\T}\log(|T(s)|) \frac{ds}{s}\right), & z_0 \in (0,1),
\end{cases}
\end{align}
and
\begin{align} \nn
T_1(z_0) &= \frac{\partial}{\partial z} \log T(z,z_0) \Big|_{z=0}\\
&= \begin{cases}
\sum\limits_{\zeta_k\in(\zeta_0,0)} (\zeta_k^{-1} -\zeta_k), & z_0 \in (-1,0),\\
\sum\limits_{\zeta_k\in(-1,0)} (\zeta_k^{-1} -\zeta_k) +
\frac{1}{\pi\I}\int\limits_{\ol{z_0}}^{z_0}\log(|T(s)|) \frac{ds}{s^2}, & |z_0| = 1, \\
\sum\limits_{\zeta_k\in(-1,0)\cup(\zeta_0,1)} (\zeta_k^{-1} -\zeta_k) +
\frac{1}{\pi \I}\int\limits_{\T}\log(|T(s)|) \frac{ds}{s^2}, & z_0 \in (0,1).
\end{cases}
\end{align}

\begin{theorem}\label{thm:asym}
Assume \eqref{decay} and abbreviate by $c_k= -\frac{\zeta_k - \zeta_k^{-1}}{2\log(|\zeta_k|)}$
the velocity of the $k$'th soliton determined by $\re(\Phi(\zeta_k))=0$.
Then the asymptotics in the soliton region, $|n/t| \geq 1 + C/t \log(t)^2$ for some
$C>0$, are as follows.

Let $\eps > 0$ sufficiently small such that the intervals
$[c_k-\eps,c_k+\eps]$, $1\le k \le N$, are disjoint and lie inside $(-\infty,-1)\cup(1,\infty)$.

If $|\frac{n}{t} - c_k|<\eps$ for some $k$, one has
\begin{align}\nn
\prod_{j=n}^\infty (2 a(j,t)) &= T_0(z_0) \left(
\sqrt{\frac{1-\zeta_k^2 + \gam_k(n,t)}{1-\zeta_k^2 + \gam_k(n,t) \zeta_k^2}} + O(t^{-l}) \right),\\
\sum_{j=n+1}^\infty b(j,t) &= \frac{1}{2} T_1(z_0) -
\frac{\gam_k(n,t)  \zeta_k (\zeta_k^2-1)}{2((\gam_k(n,t) -1) \zeta_k^2+1)} + O(t^{-l}),
\end{align}
for any $l \geq 1$, where
\be
\gam_k(n,t) = \gam_k T(\zeta_k,-c_k - \sqrt{c_k^2 -1})^{-2} \E^{t (\zeta_k - \zeta_k^{-1})} \zeta_k^{2n}.
\ee

If $|\frac{n}{t} -c_k| \geq \eps$, for all $k$, one has
\begin{align}\nn
\prod_{j=n}^\infty (2 a(j,t)) &= T_0(z_0) \left(1 + O(t^{-l}) \right),\\
\sum_{j=n+1}^\infty b(j,t) &= \frac{1}{2} T_1(z_0) + O(t^{-l}),
\end{align}
for any $l \geq 1$.
\end{theorem}

In particular, we recover the well-known fact that the solution splits into a sum of independent solitons
where the presence of the other solitons and the radiation part corresponding to the continuous spectrum
manifests itself in phase shifts given by $T(\zeta_k,-c_k - \sqrt{c_k^2 -1})^{-2}$. Indeed, notice that
for $\zeta_k\in (-1,0)$ this term just contains product over the Blaschke factors corresponding to solitons
$\zeta_j$ with $\zeta_k<\zeta_j$. For $\zeta_k\in (0,1)$ we have the product over the Blaschke factors
corresponding to solitons $\zeta_j\in(-1,0)$, the integral over the full unit circle, plus the product over the
Blaschke factors corresponding to solitons $\zeta_j$ with $\zeta_k>\zeta_j$.

The proof will be given at the end of Section~\ref{sec:solreg}. Furthermore,
in the remaining regions the analysis in Section~\ref{sec:solreg} also shows
that the Riemann--Hilbert problem reduces to one without
solitons. In fact, away from the soliton region, the asymptotics are given by
\begin{align}\nn
\prod_{j=n}^\infty (2 a(j,t)) &= T_0(-1) \prod_{j=n}^\infty (2 \ti{a}(j,t)) \left(1 + O(t^{-l}) \right),\\
\sum_{j=n+1}^\infty b(j,t) &= \frac{1}{2} T_1(-1) + \sum_{j=n+1}^\infty \ti{b}(j,t) + O(t^{-l}),
\end{align}
where $\ti{a}(n,t)$, $\ti{b}(n,t)$ are the solutions corresponding to the case without solitons and
with $R_+(z,0)$ replaced by
\be
\ti{R}_+(z,0) = T(z,-1)^{-2} R_+(z,0).
\ee
Note that the Blaschke product
\[
T(z,-1)= \prod\limits_{\zeta_k\in(-1,0)} |\zeta_k| \frac{z-\zeta_k^{-1}}{z-\zeta_k}
\]
satisfies $|T(z,-1)|=1$ for $z\in\T$.
Hence everything is reduced to the case studied in \cite{km}.

Finally we remark that the same method can be used to handle solitons on a periodic
background \cite{krt} (cf.\ also \cite{emt4}, \cite{kt}, \cite{kt2}).

\section{The Inverse scattering transform and the Riemann--Hilbert problem}
\label{sec:istrhp}

In this section we want to derive the Riemann--Hilbert problem from scattering theory.
The special case without eigenvalues was first given in Kamvissis \cite{km}. The
eigenvalues will be added by appropriate pole conditions which are then
turned into jumps following Deift, Kamvissis, Kriecherbauer, and Zhou \cite{dkkz}.

For the necessary results from scattering theory respectively the inverse
scattering transform for the Toda lattice we refer to \cite{tist}, \cite{tivp}, \cite{tjac}.

Associated with $a(t), b(t)$ is a self-adjoint Jacobi operator
\begin{equation} \label{defjac}
H(t) = a(t) S^+  + a^-(t) S^-  + b(t)
\end{equation}
in $\lz$, where $S^\pm f(n) = f^\pm(n)= f(n\pm1)$ are the usual shift operators and
$\lz$ denotes the Hilbert space of square summable (complex-valued) sequences
over $\Z$. By our assumption \eqref{decay} the spectrum of $H$ consists of an absolutely
continuous part $[-1,1]$ plus a finite number of eigenvalues $\lam_k\in\R\backslash[-1,1]$,
$1\le k \le N$. In addition, there exist two Jost functions $\psi_\pm(z,n,t)$
which solve the recurrence equation
\be
H(t) \psi_\pm(z,n,t) = \frac{z+z^{-1}}{2} \psi_\pm(z,n,t), \qquad |z|\le 1,
\ee
and asymptotically look like the free solutions
\be
\lim_{n \to \pm \infty} z^{\mp  n} \psi_{\pm}(z,n,t) =1.
\ee
Both $\psi_\pm(z,n,t)$ are analytic for $0<|z|<1$ with smooth boundary values
for $|z|=1$.
The asymptotics of the two projections of the Jost function are
\be\label{eq:psiasym}
\psi_\pm(z,n,t) =  \frac{z^{\pm n}}{A_\pm(n,t)} \Big(1 + 2 B_\pm(n,t) z + O(z^2) \Big),
\ee
as $z \to 0$, where
\be \label{defAB}
\aligned
A_+(n,t) &= \prod_{j=n}^{\infty} 2 a(j,t), \quad
B_+(n,t)= -\sum_{j=n+1}^\infty b(j,t), \\
A_-(n,t) &= \!\!\prod_{j=- \infty}^{n-1}\! 2 a(j,t), \quad
B_-(n,t) = -\sum_{j=-\infty}^{n-1} b(j,t).
\endaligned
\ee

One has the scattering relations
\be \label{relscat}
T(z) \psi_\mp(z,n,t) =  \ol{\psi_\pm(z,n,t)} +
R_\pm(z,t) \psi_\pm(z,n,t),  \qquad |z|=1,
\ee
where $T(z)$, $R_\pm(z,t)$ are the transmission respectively reflection coefficients.
The transmission and reflection coefficients have the following well-known properties:

\begin{lemma}
The transmission coefficient $T(z)$ has a meromorphic extension to the
interior of the unit circle with simple poles at the images of the eigenvalues $\zeta_k$.
The residues of $T(z)$ are given by
\be\label{eq:resT}
\res_{\zeta_k} T(z) = - \zeta_k \frac{\gam_{+,k}(t)}{\mu_k(t)} = - \zeta_k \gam_{-,k}(t) \mu_k(t),
\ee
where
\be
\gam_{\pm,k}(t)^{-1} = \sum_{n\in\Z} |\psi_\pm(\zeta_k,n,t)|^2
\ee
and $\psi_- (\zeta_k,n,t) = \mu_k(t) \psi_+(\zeta_k,n,t)$.

Moreover,
\be \label{reltrpm}
T(z) \ol{R_+(z,t)} + \ol{T(z)} R_-(z,t)=0, \qquad |T(z)|^2 + |R_\pm(z,t)|^2=1.
\ee
\end{lemma}

In particular one reflection coefficient, say $R(z,t)=R_+(z,t)$, and one set of
norming constants, say $\gam_k(t)= \gam_{+,k}(t)$, suffices. Moreover,
the time dependence is given by:

\begin{lemma}
The time evolutions of the quantities $R_+(z,t)$, $\gam_{+,k}(t)$ are given by
\begin{align}
R(z,t) &= R(z) \E^{t (z - z^{-1})}\\
\gam_k(t) &= \gam_k \E^{t (\zeta_k - \zeta_k^{-1})},
\end{align}
where $R(z)=R(z,0)$ and $\gam_k=\gam_k(0)$.
\end{lemma}

Now we define the sectionally meromorphic vector
\be\label{defm}
m(z,n,t)= \left\{\begin{array}{c@{\quad}l}
\begin{pmatrix} T(z) \psi_-(z,n,t) z^n  & \psi_+(z,n,t) z^{-n} \end{pmatrix},
& |z|<1,\\
\begin{pmatrix} \psi_+(z^{-1},n,t) z^n & T(z^{-1}) \psi_-(z^{-1},n,t) z^{-n} \end{pmatrix},
& |z|>1.
\end{array}\right.
\ee
We are interested in the jump condition of $m(z,n,t)$ on the unit circle $\T$ (oriented
counterclockwise). To formulate our jump condition we use the following convention:
When representing functions on $\T$, the lower subscript denotes
the non-tangential limit from different sides,
\be
m_\pm(z) = \lim_{ \zeta\to z,\; |\zeta|^{\pm 1}<1} m(\zeta), \qquad |z|=1.
\ee
In general, for an oriented contour $\Sigma$, $m_+(z)$ (resp.\ $m_-(z)$) will denote the limit
of $m(\zeta)$ as $\zeta\to z$ from the positive (resp.\ negative) side of $\Sigma$.
Using the notation above implicitly assumes that these limits exist in the sense that
$m(z)$ extends to a continuous function on the boundary.

\begin{theorem}[Vector Riemann--Hilbert problem]\label{thm:vecrhp}
Let $\mathcal{S}_+(H(0))=\{ R(z),\; |z|=1; \: (\zeta_k, \gam_k), \: 1\le k \le N \}$
the right scattering data of the operator $H(0)$. Then $m(z)=m(z,n,t)$ defined in \eqref{defm}
is meromorphic away from the unit circle with simple poles at $\zeta_k$, $\zeta_k^{-1}$ and satisfies:
\begin{enumerate}
\item The jump condition
\be \label{eq:jumpcond}
m_+(z)=m_-(z) v(z), \qquad
v(z)=\begin{pmatrix}
1-|R(z)|^2 & - \ol{R(z)} \E^{-t\Phi(z)} \\
R(z) \E^{t\Phi(z)} & 1
\end{pmatrix},
\ee
for $z \in\T$,
\item
the pole conditions
\be\label{eq:polecond}
\aligned
\res_{\zeta_k} m(z) &= \lim_{z\to\zeta_k} m(z)
\begin{pmatrix} 0 & 0\\ - \zeta_k \gam_k \E^{t\Phi(\zeta_k)}  & 0 \end{pmatrix},\\
\res_{\zeta_k^{-1}} m(z) &= \lim_{z\to\zeta_k^{-1}} m(z)
\begin{pmatrix} 0 & \zeta_k^{-1} \gam_k \E^{t\Phi(\zeta_k)} \\ 0 & 0 \end{pmatrix},
\endaligned
\ee
\item
the symmetry condition
\be \label{eq:symcond}
m(z^{-1}) = m(z) \sigI
\ee
\item
and the normalization
\be\label{eq:normcond}
m(0) = (m_1\quad m_2),\quad m_1 \cdot m_2 = 1\quad m_1 > 0.
\ee
\end{enumerate}
Here the phase is given by
\begin{equation}
\Phi(z)=z-z^{-1}+2 \frac{n}{t} \log \, z.
\end{equation}
\end{theorem}

\begin{proof}
The jump condition \eqref{eq:jumpcond} is a simple calculation using the scattering relations
\eqref{relscat} plus \eqref{reltrpm}. The pole conditions follow since $T(z)$ is meromorphic in $|z| <1$
with simple poles at $\zeta_k$ and residues given by \eqref{eq:resT}.
The symmetry condition holds by construction and the normalization \eqref{eq:normcond}
is immediate from the following lemma.
\end{proof}

Observe that the pole condition at $\zeta_k$ is sufficient since the one at $\zeta_k^{-1}$ follows
by symmetry.

Moreover, we have the following asymptotic behaviour near $z=0$:

\begin{lemma}
The function $m(z,n,t)$ defined in \eqref{defm} satisfies
\be\label{eq:AB}
m(z,n,t) = \begin{pmatrix}
A(n,t) (1 - 2 B(n-1,t) z) &
\frac{1}{A(n,t)}(1 + 2 B(n,t) z )
\end{pmatrix} + O(z^2).
\ee
Here $A(n,t)= A_+(n,t)$ and $B(n,t)= B_+(n,t)$ are defined in \eqref{defAB}.
\end{lemma}

\begin{proof}
This follows from \eqref{eq:psiasym} and $T(z)= A_+ A_- ( 1 - 2(B_+ - b +B_-)z+ O(z^2))$.
\end{proof}

For our further analysis it will be convenient to rewrite the pole condition as a jump
condition and hence turn our meromorphic Riemann--Hilbert problem into a holomorphic Riemann--Hilbert problem following \cite{dkkz}.
Choose $\eps$ so small that the discs $|z-\zeta_k|<\eps$ are inside the unit circle and
do not intersect. Then redefine $m$ in a neighborhood of $\zeta_k$ respectively $\zeta_k^{-1}$ according to
\be\label{eq:redefm}
m(z) = \begin{cases} m(z) \begin{pmatrix} 1 & 0 \\
\frac{\zeta_k \gamma_k \E^{t\Phi(\zeta_k)} }{z-\zeta_k} & 1 \end{pmatrix},  &
|z-\zeta_k|< \eps,\\
m(z) \begin{pmatrix} 1 & -\frac{z \gamma_k \E^{t\Phi(\zeta_k)} }{z-\zeta_k^{-1}} \\
0 & 1 \end{pmatrix},  &
|z^{-1}-\zeta_k|< \eps,\\
m(z), & \text{else}.\end{cases}
\ee
Then a straightforward calculation using $\res_\zeta m = \lim_{z\to\zeta} (z-\zeta)m(z)$ shows

\begin{lemma}\label{lem:pctoj}
Suppose $m(z)$ is redefined as in \eqref{eq:redefm}. Then $m(z)$ is holomorphic away from
the unit circle and satisfies \eqref{eq:jumpcond}, \eqref{eq:symcond}, \eqref{eq:normcond}
and the pole conditions are replaced by the jump conditions
\be \label{eq:jumpcond2}
\aligned
m_+(z) &= m_-(z) \begin{pmatrix} 1 & 0 \\
\frac{\zeta_k \gamma_k \E^{t\Phi(\zeta_k)}}{z-\zeta_k} & 1 \end{pmatrix},\quad |z-\zeta_k|=\eps,\\
m_+(z) &= m_-(z) \begin{pmatrix} 1 & \frac{z \gamma_k \E^{t\Phi(\zeta_k)}}{z-\zeta_k^{-1}} \\
0 & 1 \end{pmatrix},\quad |z^{-1}-\zeta_k|=\eps,
\endaligned
\ee
where the small circle around $\zeta_k$ is oriented counterclockwise and the one around
$\zeta_k^{-1}$ is oriented clockwise.
\end{lemma}

Next we turn to uniqueness of the solution of this vector Riemann--Hilbert problem. This will also explain the
reason for our symmetry condition. We begin by observing that if
there is a point $z_1\in\C$, such that $m(z_1)=\rN$, then $n(z)=\frac{1}{z-z_1} m(z)$
satisfies the same jump and pole conditions as $m(z)$. However, it will clearly
violate the symmetry condition! Hence, without the symmetry condition, the solution
of our vector Riemann--Hilbert problem will not be unique in such a situation. Moreover, a look at the
one soliton solution verifies that this case indeed can happen.

\begin{lemma}[One soliton solution]\label{lem:singlesoliton}
Suppose there is only one eigenvalue and a vanishing reflection coefficient, that is,
$\mathcal{S}_+(H(t))=\{ R(z)\equiv 0,\; |z|=1; \: (\zeta, \gam) \}$ with $\zeta\in(-1,0)\cup(0,1)$ and $\gam\ge0$.
Then the Riemann--Hilbert problem \eqref{eq:jumpcond}--\eqref{eq:normcond} has a unique solution
is given by
\begin{align}\label{eq:oss}
m_0(z) &= \begin{pmatrix} f(z) & f(1/z) \end{pmatrix} \\
\nn f(z) &= \frac{1}{\sqrt{1 - \zeta^2 + \gam(n,t)} \sqrt{1 - \zeta^2 + \zeta^2 \gam(n,t)}}
\left(\gam(n,t) \zeta^2 \frac{z-\zeta^{-1}}{z - \zeta} + 1 - \zeta^2\right),
\end{align}
where $\gam(n,t)=\gam \E^{t\Phi(\zeta)}$.
In particular,
\be
A_+(n,t) = \sqrt{\frac{1-\zeta^2 + \gam(n,t)}{1 - \zeta^2 + \gam(n,t) \zeta^2}}, \qquad
B_+(n,t) = \frac{\gam(n,t)  \zeta (\zeta ^2-1)}{2 (1 - \zeta^2 + \gam(n,t) \zeta^2)}.
\ee
\end{lemma}

\begin{proof}
By symmetry, the solution must be of the form $m_0(z) = \begin{pmatrix} f(z) & f(1/z) \end{pmatrix}$,
where $f(z)$ is meromorphic in $\C\cup\{\infty\}$ with the only possible pole at $\zeta$. Hence
\[
f(z) = \frac{1}{A} \left( 1+ 2 \frac{B}{z - \zeta}\right),
\]
where the unknown constants $A$ and $B$ are uniquely determined by the pole condition
$\res_\zeta f(z) = -\zeta \gam(n,t) f(\zeta^{-1})$ and the normalization $f(0) f(\infty)=1$, $f(0)>0$.
\end{proof}

In fact, observe $f(z_1)=f(z_1^{-1})=0$ if and only if $z_1=\pm 1$ and $\gam=\pm(\zeta^{-1}-\zeta)$.
Furthermore, even in the general case $m(z_1)=\rN$ can only occur at $z_1=\pm1$ as the
following lemma shows.

\begin{lemma}\label{lem:resonant}
If $m(z_1) = \rN$ for $m$ defined as in \eqref{defm}, then $z_1  = \pm 1$. Moreover,
the zero of at least one component is simple in this case.
\end{lemma}

\begin{proof}
By \eqref{defm} the condition $m(z_1) = \rN$ implies that either the Jost solutions $\psi_-(z_1,n)$ and
$\psi_+(z_1,n)$ are linearly dependent or $T(z_1)=0$. This can only happen, at a band edge,
$z_1 = \pm 1$, or at an eigenvalue $z_1=\zeta_j$.

We begin with the case $z_1=\zeta_j$. In this case the derivative of the Wronskian
$W(z)=a(n)(\psi_+(z,n)\psi_-(z,n+1)-\psi_+(z,n+1)\psi_-(z,n))$ does not vanish
$\frac{d}{dz} W(z) |_{z=z_1} \ne 0$ (\cite[Chap.~10]{tjac}). Moreover,
the diagonal Green's function $g(\lam,n)= W(z)^{-1} \psi_+(z,n) \psi_-(z,n)$ is
Herglotz and hence can have at most a simple zero at $z=z_1$. Hence, if
$\psi_+(\zeta_j,n) = \psi_-(\zeta_j,n) =0$, both can have at most a simple zero at $z=\zeta_j$.
But $T(z)$ has a simple pole at $\zeta_j$ and hence $T(z) \psi_-(z,n)$ cannot
vanish at $z=\zeta_j$, a contradiction.

It remains to show that one zero is simple in the case $z_1=\pm 1$. In fact,
one can show that $\frac{d}{dz} W(z) |_{z=z_1} \ne 0$ in this case as follows:
First of all note that $\psi_\pm'(z)$ (where $\prime$ denotes the derivative with respect to
$z$) again solves $H\psi_\pm'(z_1) = \lam_1 \psi_\pm'(z_1)$ if $z_1=\pm1$. Moreover, by
$W(z_1)=0$ we have $\psi_+(z_1) = c \psi_-(z_1)$ for some constant $c$ (independent of $n$).
Thus we can compute
\begin{align*}
W'(z_1) &= W(\psi_+'(z_1),\psi_-(z_1)) + W(\psi_+(z_1),\psi_-'(z_1))\\
&= c^{-1} W(\psi_+'(z_1),\psi_+(z_1)) + c W(\psi_-(z_1),\psi_-'(z_1))
\end{align*}
by letting $n\to+\infty$ for the first and $n\to-\infty$ for the second Wronskian (in which case we can
replace $\psi_\pm(z_1)$ by $z_1^{\pm n}$),
which gives
\[
W'(z_1) = \frac{c+c^{-1}}{2}.
\]
Hence the Wronskian has a simple zero. But if both functions had more than
simple zeros, so would the Wronskian, a contradiction.
\end{proof}

Finally, it is interesting to note that the assumptions $\zeta\in(-1,0)\cup(0,1)$ and $\gam\ge0$ are
crucial for uniqueness. Indeed, if we choose $\gam= \zeta^2-1<0$, then every solution
is a multiple of $f(z)=z\zeta^{-1}(z-\zeta)^{-1}$ which cannot be normalized at $0$.

\section{A uniqueness result for symmetric vector Riemann--Hilbert problems}

In this section we want to investigate uniqueness for the holomorphic vector Riemann--Hilbert problem
\begin{align}\nn
& m_+(z) = m_-(z) v(z), \qquad z\in \Sigma,\\ \label{eq:rhp4m}
& m(z^{-1}) = m(z) \sigI,\\ \nn
& m(0) = \begin{pmatrix} 1 & m_2\end{pmatrix}.
\end{align}
where $\Sigma$ is a nice oriented contour (see Hypothesis~\ref{hyp:sym}), symmetric with respect to
$z\mapsto z^{-1}$, and $v$ is continuous satisfying $\det(v)=1$ and
\be
v(z^{-1}) = \sigI v(z)^{-1} \sigI,\quad z\in\Sigma.
\ee
The normalization used here will be more convenient than \eqref{eq:normcond}.
In fact, \eqref{eq:normcond} will be satisfied by $m_2^{-1/2} m(z)$.

Now we are ready to show that the symmetry condition in fact guarantees uniqueness.

\begin{theorem}
Suppose there exists a solution $m(z)$ of the Riemann--Hilbert problem \eqref{eq:rhp4m} for which
$m(z)=\begin{pmatrix} 0 & 0\end{pmatrix}$ can happen at most for $z=\pm1$ in which case
$\limsup_{z\to\pm 1} \frac{z\mp 1}{m_j(z)}$ is bounded from any direction for $j=1$ or $j=2$.

Then the Riemann--Hilbert problem \eqref{eq:rhp4m} with norming condition replaced by
\be\label{eq:rhp4ma}
m(0) = \begin{pmatrix} \alpha & m_2\end{pmatrix}
\ee
for given $\alpha\in\C$, has a unique solution $m_\alpha(z) = \alpha\, m(z)$.
\end{theorem}

\begin{proof}
Let $m_\alpha(z)$ be a solution of \eqref{eq:rhp4m} normalized according to
\eqref{eq:rhp4ma}. Then we can construct a matrix valued solution via $M=(m, m_\alpha)$ and
there are two possible cases: Either $\det M(z)$ is nonzero for some $z$ or it vanishes
identically.

We start with the first case. Since the determinant of our Riemann--Hilbert problem has no jump
and is bounded at infinity, it is constant. But taking determinants in
\[
M(z^{-1}) = M(z) \sigI.
\]
gives a contradiction.

It remains to investigate the case where $\det(M)\equiv 0$. In this case
we have $m_\alpha(z) = \delta(z) m(z)$ with a scalar function $\delta$. Moreover,
$\delta(z)$ must be holomorphic for $z\in\C\backslash\Sigma$ and continuous
for $z\in\Sigma$ except possibly at the points where $m(z_0) = \rN$. Since it has
no jump across $\Sigma$,
\[
\delta_+(z) m_+(z) = m_{\alpha,+}(z) = m_{\alpha,-}(z) v(z) = \delta_-(z) m_-(z) v(z)
= \delta_-(z) m_+(z),
\]
it is even holomorphic in $\C\backslash\{\pm 1\}$ with at most
simple poles at $z=\pm 1$. Hence it must be of the form
\[
\delta(z) = A + \frac{B}{z-1} + \frac{C}{z +1}.
\]
Since $\delta$ has to be symmetric, $\delta(z) = \delta(z^{-1})$, we obtain $B = C = 0$. Now, by
the normalization we obtain $\delta(z) = A = \alpha$. This finishes the proof.
\end{proof}

Furthermore, note that the requirements cannot be relaxed to allow (e.g.) second order
zeros in stead of simple zeros. In fact, if $m(z)$ is a solution for which both components
vanish of second order at, say, $z=+1$, then $\ti{m}(z)=\frac{z}{(z-1)^2} m(z)$ is a
nontrivial symmetric solution of the vanishing problem (i.e.\ for $\alpha=0$).

By Lemma~\ref{lem:resonant} we have

\begin{corollary}\label{cor:unique}
The function $m(z,n,t)$ defined in \eqref{defm} is the only solution of the
vector Riemann--Hilbert problem \eqref{eq:jumpcond}--\eqref{eq:normcond}.
\end{corollary}

Observe that there is nothing special about the point $z=0$ where we normalize, any
other point would do as well. However, observe that for the one soliton solution \eqref{eq:oss},
$f(z)$ vanishes at
\[
z = \zeta \frac{1+\gam-\zeta^2}{(\gam -1) \zeta ^2+1}
\]
and hence the Riemann--Hilbert problem normalized at this point has a nontrivial solution for $\alpha=0$ and
hence, by our uniqueness result, no solution for $\alpha=1$. This shows that
uniqueness and existence connected, a fact which is not surprising since our
Riemann--Hilbert problem is equivalent to a singular integral equation which is Fredholm of index
zero (see Appendix~\ref{sec:sieq}).

\section{Solitons and the soliton region}
\label{sec:solreg}

This section demonstrates the basic method of passing from a Riemann--Hilbert problem
involving solitons to one without. Furthermore, the asymptotics inside the soliton region
are computed. Solitons are represented in a Riemann--Hilbert problem by pole conditions,
for this reason we will further study how poles can be dealt with in this section.

For easy reference we note the following result which can be checked by a straightforward
calculation.

\begin{lemma}[Conjugation]\label{lem:conjug}
Assume that $\wti{\Sigma}\subseteq\Sigma$. Let $D$ be a matrix of the form
\be
D(z) = \begin{pmatrix} d(z)^{-1} & 0 \\ 0 & d(z) \end{pmatrix},
\ee
where $d: \C\backslash\wti{\Sigma}\to\C$ is a sectionally analytic function. Set
\be
\ti{m}(z) = m(z) D(z),
\ee
then the jump matrix transforms according to
\be
\ti{v}(z) = D_-(z)^{-1} v(z) D_+(z).
\ee
If $d$ satisfies $d(z^{-1}) = d(z)^{-1}$ and $d(0) > 0$. Then the transformation $\ti{m}(z) = m(z) D(z)$
respects our symmetry, that is, $\ti{m}(z)$ satisfies \eqref{eq:symcond} if and only if $m(z)$ does.
\end{lemma}

In particular, we obtain
\be
\ti{v} = \begin{pmatrix} v_{11} & v_{12} d^{2} \\ v_{21} d^{-2}  & v_{22} \end{pmatrix},
\qquad z\in\Sigma\backslash\wti{\Sigma},
\ee
respectively
\be
\ti{v} = \begin{pmatrix} \frac{d_-}{d_+} v_{11} & v_{12} d_+ d_- \\
v_{21} d_+^{-1} d_-^{-1}  & \frac{d_+}{d_-} v_{22} \end{pmatrix},
\qquad z\in\Sigma\cap\wti{\Sigma}.
\ee

In order to remove the poles there are two cases to distinguish. If
$\lam_k >\frac{1}{2}(\zeta_0+\zeta_0^{-1})$ the jump is exponentially decaying and there is nothing
to do.

Otherwise we use conjugation to turn the jumps into exponentially decaying
ones, again following Deift, Kamvissis, Kriecherbauer, and Zhou \cite{dkkz}.
It turns out that we will have to handle the poles at $\zeta_k$ and $\zeta_k^{-1}$
in one step in order to preserve symmetry and in order to not add additional poles
elsewhere.

\begin{lemma}\label{lem:twopolesinc}
Assume that the Riemann--Hilbert problem for $m$ has jump conditions near $\zeta$ and
$\zeta^{-1}$ given by
\be
\aligned
m_+(z)&=m_-(z)\begin{pmatrix}1&0\\ \frac{\gam \zeta}{z-\zeta}&1\end{pmatrix}, && |z-\zeta|=\eps, \\
m_+(z)&=m_-(z)\begin{pmatrix}1& \frac{\gam z}{z-\zeta^{-1}}\\0&1\end{pmatrix}, && |z^{-1}- \zeta|=\eps.
\endaligned
\ee
Then this Riemann--Hilbert problem is equivalent to a Riemann--Hilbert problem for $\ti{m}$ which has jump conditions near $\zeta$ and
$\zeta^{-1}$ given by
\begin{align*}
\ti{m}_+(z)&= \ti{m}_-(z)\begin{pmatrix}1& \frac{(\zeta z-1)^2}{\zeta (z-\zeta) \gam}\\ 0 &1\end{pmatrix},
&& |z-\zeta|=\eps, \\
\ti{m}_+(z)&= \ti{m}_-(z)\begin{pmatrix}1& 0 \\ \frac{(z-\zeta)^2}{\zeta z (\zeta z-1) \gam}&1\end{pmatrix},
&& |z^{-1}- \zeta|=\eps,
\end{align*}
and all remaining data conjugated (as in Lemma~\ref{lem:conjug}) by
\be
D(z) = \begin{pmatrix} \frac{z - \zeta}{\zeta z-1} & 0 \\ 0 & \frac{\zeta z-1}{z-\zeta} \end{pmatrix}.
\ee
\end{lemma}

\begin{proof}
To turn $\gam$ into $\gam^{-1}$, introduce $D$ by
\[
D(z) = \begin{cases}
\begin{pmatrix} 1 & \frac{1}{\gam} \frac{z-\zeta}{\zeta}\\ - \gam \frac{\zeta}{z-\zeta} & 0 \end{pmatrix}
\begin{pmatrix} \frac{z - \zeta}{\zeta z-1} & 0 \\ 0 & \frac{\zeta z-1}{z-\zeta} \end{pmatrix}, &  |z-\zeta|<\eps, \\
\begin{pmatrix} 0 & \gam \frac{z \zeta}{z \zeta -1} \\ -\frac{1}{\gam} \frac{z \zeta -1}{z \zeta} & 1 \end{pmatrix}
\begin{pmatrix} \frac{z - \zeta}{\zeta z-1} & 0 \\ 0 & \frac{\zeta z-1}{z-\zeta} \end{pmatrix}, & |z^{-1}-\zeta|<\eps, \\
\begin{pmatrix} \frac{z - \zeta}{\zeta z-1} & 0 \\ 0 & \frac{\zeta z-1}{z-\zeta} \end{pmatrix}, & \text{else},
\end{cases}
\]
and note that $D(z)$ is analytic away from the two circles. Now set $\ti{m}(z) = m(z) D(z)$, which is again
symmetric by $D(z^{-1})= \ssigI D(z) \ssigI$.
The jumps along $|z-\zeta|=\eps$ and $|z^{-1}- \zeta|=\eps$ follow by a straightforward calculation and
the remaining jumps follow from Lemma~\ref{lem:conjug}.
\end{proof}

Now we are ready to prove our main result:

\begin{proof}[Proof of Theorem~\ref{thm:asym}]
We begin by observing that the partial transmission coefficient $T(z,z_0)$ introduced in \eqref{def:Tzz0}
satisfies the following scalar meromorphic Riemann--Hilbert problem:
\begin{enumerate}
\item $T(z,z_0)$ is meromorphic in $\C\backslash\Sigma(z_0)$, where $\Sigma(z_0)$ is the
arc given by $\Sigma(z_0) = \{z \in\T | \re(z)<\re(z_0)\}$, with simple poles at
$\zeta_k$ and simple zeros at $\zeta_k^{-1}$ for all $k$ with
$\lam_k < \frac{1}{2}(\zeta_0+\zeta_0^{-1})$,
\item
$T_+(z,z_0) = T_-(z,z_0) (1 - |R(z)|^2)$ for $z\in\Sigma(z_0)$,
\item
$T(z^{-1},z_0) = T(z,z_0)^{-1}$, $z\in\C\backslash\Sigma(z_0)$, and $T(0,z_0)>0$.
\end{enumerate}
Note also $\ol{T(z,z_0)}=T(\ol{z},z_0)$ and in particular $T(z,z_0)$ is real-valued for
$z\in\R$.

Next introduce
\[
D(z) = \begin{cases}
\begin{pmatrix} 1 & \frac{z-\zeta_k}{\zeta_k \gam_k \E^{t\Phi(\zeta_k)}}\\
- \frac{\zeta_k \gam_k \E^{t\Phi(\zeta_k)}}{z-\zeta_k} & 0 \end{pmatrix}
D_0(z), &  |z-\zeta_k|<\eps, \: \lam_k < \frac{1}{2}(\zeta_0+\zeta_0^{-1}),\\
\begin{pmatrix} 0 & \frac{z \zeta_k \gam_k \E^{t\Phi(\zeta_k)}}{z \zeta_k -1} \\
-\frac{z \zeta_k -1}{z \zeta_k \gam_k \E^{t\Phi(\zeta_k)}} & 1 \end{pmatrix}
D_0(z), & |z^{-1}-\zeta_k|<\eps, \: \lam_k < \frac{1}{2}(\zeta_0+\zeta_0^{-1}),\\
D_0(z), & \text{else},
\end{cases}
\]
where
\[
D_0(z) = \begin{pmatrix} T(z,z_0)^{-1} & 0 \\ 0 & T(z,z_0) \end{pmatrix}.
\]
Note that we have
\[
D(z^{-1})= \sigI D(z) \sigI.
\]
Now we conjugate our vector $m(z)$ defined in \eqref{defm} respectively \eqref{eq:redefm}
using $D(z)$, $\ti{m}(z)=m(z) D(z)$. Since $T(z,z_0)$ is either nonzero and continuous near
$z=\pm1$ (if $\pm 1 \notin\Sigma(z_0)$) or it has the same behaviour as $T(z)$ near $z=\pm 1$
(if $\pm 1 \in\Sigma(z_0)$), the new vector $\ti{m}(z)$ is again continuous near $z=\pm 1$
(even if $T(z)$ vanishes there).

Then using Lemma~\ref{lem:conjug} and Lemma~\ref{lem:twopolesinc} the jump
corresponding $\lam_k <\frac{1}{2}(\zeta_0+\zeta_0^{-1})$ (if any) is given by
\be
\aligned
\ti{v}(z) &= \begin{pmatrix}1& \frac{z-\zeta_k}{\zeta_k
\gam_k T(z,z_0)^{-2} \E^{t\Phi(\zeta_k)} }\\ 0 &1\end{pmatrix},
\qquad |z-\zeta_k|=\eps, \\
\ti{v}(z) &= \begin{pmatrix}1& 0 \\ \frac{\zeta_k z -1}{\zeta_k z \gam_k T(z,z_0)^2
\E^{t\Phi(\zeta_k)}}&1\end{pmatrix},
\qquad |z^{-1}- \zeta_k|=\eps,
\endaligned
\ee
and corresponding $\lam_k \ge \frac{1}{2}(\zeta_0+\zeta_0^{-1})$ (if any) by
\be
\aligned
\ti{v}(z) &= \begin{pmatrix} 1 & 0 \\ \frac{\zeta_k \gam_k T(z,z_0)^{-2} \E^{t\Phi(\zeta_k)}}{z-\zeta_k}
 & 1 \end{pmatrix},
\qquad |z-\zeta_k|=\eps, \\
\ti{v}(z) &= \begin{pmatrix} 1 & \frac{z \gam_k T(z,z_0)^2 \E^{t\Phi(\zeta_k)}}{z-\zeta_k^{-1}} \\
0 & 1 \end{pmatrix},
\qquad |z^{-1}-\zeta_k|=\eps.
\endaligned
\ee
In particular, an investigation of the sign of $\re(\Phi(z))$ shows
that all off-diagonal entries of these jump matrices, except for possibly one if
$\zeta_{k_0}=\zeta_0$ for some $k_0$, are exponentially decreasing. In this case we will keep the
pole condition which now reads
\be
\aligned
\res_{\zeta_{k_0}} \ti{m}(z) &= \lim_{z\to\zeta_{k_0}} \ti{m}(z)
\begin{pmatrix} 0 & 0\\ - \zeta_{k_0} \gam_{k_0} T(\zeta_{k_0},z_0)^{-2} \E^{t\Phi(\zeta_{k_0})}  & 0 \end{pmatrix},\\
\res_{\zeta_{k_0}^{-1}} \ti{m}(z) &= \lim_{z\to\zeta_{k_0}^{-1}} \ti{m}(z)
\begin{pmatrix} 0 & \zeta_{k_0}^{-1} \gam_{k_0} T(\zeta_{k_0},z_0)^{-2} \E^{t\Phi(\zeta_{k_0})} \\ 0 & 0 \end{pmatrix}.
\endaligned
\ee
Furthermore, the jump along
$\T$ is given by
\be
\ti{v}(z) = \begin{cases}
\ti{b}_-(z)^{-1} \ti{b}_+(z), \qquad \lam(z)> \lam_0,\\
\ti{B}_-(z)^{-1} \ti{B}_+(z), \qquad \lam(z)< \lam_0,\\
\end{cases}
\ee
where
\[
\ti{b}_-(z) = \begin{pmatrix} 1 & \frac{R(z^{-1}) \E^{-t\Phi(z)}}{T(z^{-1},z_0)^2} \\ 0 & 1 \end{pmatrix}, \quad
\ti{b}_+(z) = \begin{pmatrix} 1 & 0 \\ \frac{R(z) \E^{t\Phi(z)}}{T(z,z_0)^2}& 1 \end{pmatrix},
\]
and
\[
\ti{B}_-(z) = \begin{pmatrix} 1 & 0 \\ - \frac{T_-(z,z_0)^{-2}}{1-|R(z)|^2} R(z) \E^{t\Phi(z)} & 1 \end{pmatrix}\!, \quad
\ti{B}_+(z) = \begin{pmatrix} 1 & - \frac{T_+(z,z_0)^2}{1-|R(z)|^2} \ol{R(z)} \E^{-t\Phi(z)} \\ 0 & 1 \end{pmatrix}\!.
\]
Now the jump along $\T$ can also be made arbitrarily small following
the nonlinear steepest descent method developed by Deift and Zhou \cite{dz}:
Split the Fourier transform of the reflection coefficient into a part which has an analytic
extension to a neighborhood of the unit circle plus a small error. Move the analytic part
away from the unit circle using the factorizations from above. Since the Fourier
coefficients decay faster than any polynomial, the errors from both parts can be
made $O(t^{-l})$ for any $l\in\N$. We refer to \cite{dz} respectively \cite{km} for details.
Hence we can apply Theorem~\ref{thm:remcontour} as follows:

If $|\frac{n}{t} - c_k|>\eps$ for all $k$ we can choose $\gam_0=0$ and $w_0^t\equiv 0$.
Since the error between $w^t$ and $w_0^t$ is exponentially small, this proves the second
part of Theorem~\ref{thm:asym} upon comparing
\be
m(z) = \hat{m}(z) \begin{pmatrix} T(z,z_0) & 0\\ 0 & T(z,z_0)^{-1} \end{pmatrix}
\ee
with \eqref{eq:AB}. 

Otherwise, if $|\frac{n}{t} - c_k|<\eps$ for some $k$, we choose $\gam_0^t=\gam_k(n,t)$ and
$w_0^t \equiv 0$. Again we conclude that the error between $w^t$ and $w_0^t$ is exponentially small,
proving the first part of Theorem~\ref{thm:asym}.
\end{proof}

\appendix

\section{Singular integral equations}
\label{sec:sieq}

In this section we show how to transform a meromorphic vector Riemann--Hilbert problem
with simple poles at $\zeta$, $\zeta^{-1}$,
\begin{align}\nn
& m_+(z) = m_-(z) v(z), \qquad z\in \Sigma,\\ \label{eq:rhp5m}
& \res_{\zeta} m(z) = \lim_{z\to\zeta} m(z)
\begin{pmatrix} 0 & 0\\ - \zeta \gam  & 0 \end{pmatrix},\quad
\res_{\zeta^{-1}} m(z) &= \lim_{z\to\zeta^{-1}} m(z)
\begin{pmatrix} 0 & \zeta^{-1} \gam \\ 0 & 0 \end{pmatrix},\\ \nn
& m(z^{-1}) = m(z) \sigI,\\ \nn
& m(0) = \begin{pmatrix} 1 & m_2\end{pmatrix},
\end{align}
where $\zeta\in(-1,0)\cup(0,1)$ and $\gam\ge 0$, into a singular integral equation.
Since we require the symmetry condition \eqref{eq:symcond} for our Riemann--Hilbert
problems we need to adapt the usual Cauchy kernel to preserve this symmetry.
Moreover, we keep the single soliton as an inhomogeneous term which will play
the role of the leading asymptotics in our applications.

\begin{hypothesis}\label{hyp:sym}
Let $\Sigma$ consist of a finite number of smooth oriented finite curves in $\C$
which intersect at most finitely many times with all intersections being transversal.
Assume that the contour $\Sigma$ does not contain $0$, $\zeta$ and is invariant under
$z\mapsto z^{-1}$. It is oriented such that under the mapping $z\mapsto z^{-1}$
sequences converging from the positive sided to $\Sigma$ are mapped to sequences
converging to the negative side. Moreover, suppose the jump matrix $v$
can be factorized according to $v = b_-^{-1} b_+ = (\id-w_-)^{-1}(\id+w_+)$, where
$w_\pm = \pm(b_\pm-\id)$ are continuous and satisfy
\be\label{eq:wpmsym}
w_\pm(z^{-1}) = \sigI w_\mp(z) \sigI,\quad z\in\Sigma.
\ee
\end{hypothesis}

The classical Cauchy-transform
of a function $f:\Sigma\to \C$ which is square integrable is the
analytic function $C f: \C\backslash\Sigma\to\C$ given by
\be
C f(z) = \frac{1}{2\pi\I} \int_{\Sigma} \frac{f(s)}{s - z} ds,\qquad z\in\C\backslash\Sigma.
\ee
Denote the non-tangential boundary values from both sides (taken possibly
in the $L^2$-sense --- see e.g.\ \cite[eq.\ (7.2)]{deiftbook}) by $C_+ f$ respectively $C_- f$.
Then it is well-known that $C_+$ and $C_-$ are bounded operators $L^2(\Sigma)\to L^2(\Sigma)$, which satisfy $C_+ - C_- = \id$ and $C_+ C_- = 0$ (see e.g. \cite{bc}). Moreover, one has
the Plemelj--Sokhotsky formula (\cite{mu})
\[
C_\pm = \frac{1}{2} (\I H \pm \id),
\]
where
\be
H f(t) = \frac{1}{\pi} \dashint_\Sigma \frac{f(s)}{t-s} ds,\qquad t\in\Sigma,
\ee
is the Hilbert transform and $\dashint$ denotes the principal value integral.

In order to respect the symmetry condition we will restrict our attention to
the set $L^2_{s}(\Sigma)$ of square integrable functions $f:\Sigma\to\C^{2}$ such that
\be\label{eq:sym}
f(z^{-1}) = f(z) \sigI.
\ee
Clearly this will only be possible if we require our jump data to be symmetric as well (i.e.,
Hypothesis~\ref{hyp:sym} holds).

Next we introduce the Cauchy operator
\be
(C f)(z) = \frac{1}{2\pi\I} \int_\Sigma f(s) \Omega_\zeta(s,z)
\ee
acting on vector-valued functions $f:\Sigma\to\C^{2}$.
Here the Cauchy kernel is given by
\be
\Omega_{\zeta}(s,z) =
\begin{pmatrix} \frac{z-\zeta^{-1}}{s-\zeta^{-1}} \frac{1}{s-z} & 0 \\
0 & \frac{z-\zeta}{s-\zeta} \frac{1}{s-z} \end{pmatrix} ds =
\begin{pmatrix} \frac{1}{s-z} - \frac{1}{s-\zeta^{-1}} & 0 \\
0 & \frac{1}{s-z} - \frac{1}{s-\zeta} \end{pmatrix} ds,
\ee
for some fixed $\zeta\notin\Sigma$. In the case $\zeta=\infty$ we set
\be
\Omega_{\infty}(s,z) =
\begin{pmatrix} \frac{1}{s-z} - \frac{1}{s} & 0 \\
0 & \frac{1}{s-z} \end{pmatrix} ds.
\ee
and one easily checks the symmetry property:
\be\label{eq:symC}
\Omega_\zeta(1/s,1/z) = \sigI \Omega_\zeta(s,z) \sigI.
\ee
The properties of $C$ are summarized in the next lemma.

\begin{lemma}
Assume Hypothesis~\ref{hyp:sym}.
The Cauchy operator $C$ has the properties, that the boundary values
$C_\pm$ are bounded operators $L^2_s(\Sigma) \to L^2_s(\Sigma)$
which satisfy
\be\label{eq:cpcm}
C_+ - C_- = \id
\ee
and
\be\label{eq:Cnorm}
(Cf)(\zeta^{-1}) = (0\quad\ast), \qquad (Cf)(\zeta) = (\ast\quad 0).
\ee
Furthermore, $C$ restricts to $L^2_{s}(\Sigma)$, that is
\be
(C f) (z^{-1}) = (Cf)(z) \sigI,\quad z\in\C\backslash\Sigma
\ee
for $f\in L^2_{s}(\Sigma)$ and if $w_\pm$ satisfy (H.\ref{hyp:sym}) we also have
\be \label{eq:symcpm}
C_\pm(f w_\mp)(1/z) = C_\mp(f w_\pm)(z) \sigI,\quad z\in\Sigma.
\ee
\end{lemma}

\begin{proof}
Everything follows from \eqref{eq:symC} and the fact that $C$ inherits all properties from
the classical Cauchy operator.
\end{proof}

We have thus obtained a Cauchy transform with the required properties.
Following Section 7 and 8 of \cite{bc}, we can solve our Riemann--Hilbert problem using this
Cauchy operator.

Introduce the operator $C_w: L_s^2(\Sigma)\to L_s^2(\Sigma)$ by
\be
C_w f = C_+(fw_-) + C_-(fw_+),\quad f\in L^2_s(\Sigma)
\ee
and recall from Lemma~\ref{lem:singlesoliton} that the unique
solution corresponding to $v\equiv \id$ is given by
\[
m_0(z)= \begin{pmatrix} f(z) & f(\frac{1}{z}) \end{pmatrix}, \quad
f(z) = \frac{1}{1 - \zeta^2 + \gamma}
\left(\gamma \zeta^2 \frac{z-\zeta^{-1}}{z - \zeta} + 1 - \zeta^2\right)
\]
Observe that for $\gam=0$ we have $f(z)=1$ and for $\gam=\infty$ we have
$f(z)= \zeta^2 \frac{z-\zeta^{-1}}{z - \zeta}$. In particular, $m_0(z)$ is uniformly bounded away from $\zeta$
for all $\gam\in[0,\infty]$.

Then we have the next result.

\begin{theorem}\label{thm:cauchyop}
Assume Hypothesis~\ref{hyp:sym}.

Suppose $m$ solves the Riemann--Hilbert problem \eqref{eq:rhp5m}. Then
\be\label{eq:mOm}
m(z) = (1-c_0) m_0(z) + \frac{1}{2\pi\I} \int_\Sigma \mu(s) (w_+(s) + w_-(s)) \Omega_\zeta(s,z),
\ee
where
\[
\mu = m_+ b_+^{-1} = m_- b_-^{-1} \quad\mbox{and}\quad
c_0= \left( \frac{1}{2\pi\I} \int_\Sigma \mu(s) (w_+(s) + w_-(s)) \Omega_\zeta(s,0) \right)_{\!1}.
\]
Here $(m)_j$ denotes the $j$'th component of a vector.
Furthermore, $\mu$ solves
\be\label{eq:sing4muc}
(\id - C_w) \mu = (1-c_0) m_0(z).
\ee

Conversely, suppose $\ti{\mu}$ solves
\be\label{eq:sing4mu}
(\id - C_w) \ti{\mu} = m_0(z),
\ee
and
\[
\ti{c}_0= \left( \frac{1}{2\pi\I} \int_\Sigma \ti{\mu}(s) (w_+(s) + w_-(s)) \Omega_\zeta(s,0) \right)_{\!1} \ne -1,
\]
then $m$ defined via \eqref{eq:mOm}, with $(1-c_0)=(1+\ti{c}_0)^{-1}$ and $\mu=(1+\ti{c}_0)^{-1}\ti{\mu}$,
solves the Riemann--Hilbert problem \eqref{eq:rhp5m} and $\mu= m_\pm b_\pm^{-1}$.
\end{theorem}

\begin{proof}
This can be shown as in the non-symmetric case (cf.\ e.g.\ \cite{deiftbook}).
\end{proof}

Note that in the special case $\gamma=0$ we have $m_0(z)= \rI$ and
we can choose $\zeta$ as we please, say $\zeta=\infty$ such that $c_0=\ti{c}_0=0$
in the above theorem.

Hence we have a formula for the solution of our Riemann--Hilbert problem $m(z)$ in terms of
$(\id - C_w)^{-1} m_0$ and this clearly raises the question of bounded
invertibility of $\id - C_w$. This follows from Fredholm theory (cf.\ e.g. \cite{zh}):

\begin{lemma}
Assume Hypothesis~\ref{hyp:sym}.
The operator $\id-C_w$ is Fredholm of index zero,
\be
\ind(\id-C_w) =0.
\ee
\end{lemma}

By the Fredholm alternative, it follows that to show the bounded invertibility of $\id-C_w$
we only need to show that $\ker (\id-C_w) =0$. The latter being equivalent to
unique solvability of the corresponding vanishing Riemann--Hilbert problem in the case $\gam=0$
(where we can choose $\zeta=\infty$ such that $c_0=\ti{c}_0=0$).

\begin{corollary}
Assume Hypothesis~\ref{hyp:sym}.
A unique solution of the Riemann--Hilbert problem \eqref{eq:rhp5m} with $\gam=0$ exists if and only if
the corresponding vanishing Riemann--Hilbert problem, where the normalization condition is replaced by
$m(0)= \begin{pmatrix} 0 & m_2\end{pmatrix}$, with $m_2$ arbitrary, has at most one solution.
\end{corollary}

We are interested in comparing two Riemann--Hilbert problems associated with
respective jumps $w_0$ and $w$ with $\|w-w_0\|_\infty$ small,
where
\be
\|w\|_\infty= \|w_+\|_{L^\infty(\Sigma)} + \|w_-\|_{L^\infty(\Sigma)}.
\ee
For such a situation we have the following result:

\begin{theorem}\label{thm:remcontour}
Assume that for some data $\zeta_0$, $\gam_0^t$, $w_0^t$ the operator
\be
\id-C_{w_0^t}: L^2_s(\Sigma) \to L^2_s(\Sigma)
\ee
has a bounded inverse, where the bound is independent of $t$, and
suppose that the corresponding $\ti{c}_{0,0}^t$ is away from $-1$ again
uniformly in $t$.

Furthermore, let $\zeta=\zeta_0$, $\gam^t=\gam_0^t$ and assume $w^t$ satisfies
\be
\|w^t - w_0^t\|_\infty \leq \alpha(t)
\ee
for some function $\alpha(t) \to 0$ as $t\to\infty$. Then
$(\id-C_{w^t})^{-1}: L^2_s(\Sigma)\to L^2_s(\Sigma)$ also exists
for sufficiently large $t$ and the associated solutions of
the Riemann--Hilbert problems \eqref{eq:rhp5m} only differ by $O(\alpha(t))$
uniformly in $z$ away from $\Sigma \cup \{\zeta,\zeta^{-1}\}$.
\end{theorem}

\begin{proof}
By boundedness of the Cauchy transform, one has
\[
\|(C_{w^t} - C_{w_0^t})\| \leq const \|w\|_\infty.
\]
Thus, by the second resolvent identity, we infer that $(\id-C_{w^t})^{-1}$ exists for large $t$ and
\[
\|(\id-C_{w^t})^{-1}-(\id-C_{w_0^t})^{-1}\| = O(\alpha(t)).
\]
From which the claim follows since this implies $|\ti{c}_0^t - \ti{c}_{0,0}^t| = O(\alpha(t))$
respectively $\|\mu^t - \mu_0^t\|_{L^2} = O(\alpha(t))$ and thus
$m^t(z) - m_0^t(z) = O(\alpha(t))$ uniformly in $z$ away from $\Sigma \cup \{\zeta,\zeta^{-1}\}$.
\end{proof}

\noindent
{\bf Acknowledgments.}
We thank S.\ Kamvissis for several helpful discussions and I.\ Egorova, K.\ Grunert and A.\ Mikikits-Leitner
for pointing out errors in a previous version of this article.

\end{document}